\documentclass{article}
\usepackage{nips10submit_e,times}

\usepackage{amsmath,amssymb,amsthm}
\usepackage[usenames,dvipsnames]{pstricks}
\usepackage{epsfig}
\usepackage{pst-grad}
\usepackage{pst-plot}
\usepackage{comment}

\title{Random Projection Trees Revisited}

\author{
Aman Dhesi\thanks{Work done as an undergraduate student at IIT Kanpur}\\
Department of Computer Science\\
Princeton University\\
Princeton, New Jersey, USA.\\
\texttt{adhesi@princeton.edu} \\
\And
Purushottam Kar\\
Department of Computer Science and Engineering\\
Indian Institute of Technology\\
Kanpur, Uttar Pradesh, INDIA.\\
\texttt{purushot@cse.iitk.ac.in} \\
}

\newcommand{\M}{{\cal M}}
\newcommand{\R}{{\mathbb R}}

\renewcommand{\Pr}[1]{{\mathbb P}\left[{#1}\right]}
\newcommand{\N}[2]{{\cal N}\left({#1},{#2}\right)}

\newcommand{\RP}{{\sc RPTree}}
\newcommand{\RPM}{{\sc RPTree-Max}}
\newcommand{\RPm}{{\sc RPTree-Mean}}
\newcommand{\OO}[1]{{\cal O}\left({#1}\right)}
\newcommand{\E}[1]{{\mathbb E}\left[{#1}\right]}

\newcommand{\parr}[1]{{#1}_\parallel}
\newtheorem{lem}{Lemma}
\newtheorem{thm}[lem]{Theorem}

\newtheorem{defn}[lem]{Definition}
\newtheorem{fct}[lem]{Fact}

\nipsfinalcopy

\begin{document}

\maketitle

\begin{abstract}
The Random Projection Tree (\RP) structures proposed in \cite{rptree} are space partitioning data structures that automatically adapt to various notions of intrinsic dimensionality of data. We prove new results for both the \RPM\ and the \RPm\ data structures. Our result for \RPM\ gives a near-optimal bound on the number of levels required by this data structure to reduce the size of its cells by a factor $s \geq 2$. We also prove a packing lemma for this data structure. Our final result shows that low-dimensional manifolds have bounded Local Covariance Dimension. As a consequence we show that \RPm\ adapts to manifold dimension as well.
\end{abstract}

\section{Introduction}
\label{intro}
The Curse of Dimensionality \cite{curse} has inspired research in several directions in Computer Science and has led to the development of several novel techniques such as dimensionality reduction, sketching etc. Almost all these techniques try to map data to lower dimensional spaces while approximately preserving useful information. However, most of these techniques do not assume anything about the data other than that they are are imbedded in some high dimensional Euclidean space endowed with some distance/similarity function.

As it turns out, in many situations, the data is not simply scattered in the Euclidean space in a random fashion. Often, generative processes impose (non-linear) dependencies on the data that restrict the degrees of freedom available and result in the data having low intrinsic dimensionality. There exist several formalizations of this concept of intrinsic dimensionality. \cite{rptree} provides an excellent example of automated motion capture in which a large number of points on the body of an actor are sampled through markers and their coordinates transferred to an animated avatar. Now, although a large sample of points is required to ensure a faithful recovery of all the motions of the body (which causes each captured frame to lie in a very high dimensional space), these points are nevertheless constrained by the degrees of freedom offered by the human body which are very few.

Algorithms that try to exploit such non-linear structure in data have been studied extensively resulting in a large number of \emph{Manifold Learning} algorithms for example \cite{isomap,manifold-jl-naor,manifold-jl-baraniuk}. These techniques typically assume knowledge about the manifold itself or the data distribution. For example, \cite{manifold-jl-naor} and \cite{manifold-jl-baraniuk} require knowledge about the intrinsic dimensionality of the manifold. \cite{isomap} requires a sampling of points that is ``sufficiently'' dense with respect to some manifold parameters.

Recently in \cite{rptree}, Dasgupta and Freund proposed space partitioning algorithms that adapt to the intrinsic dimensionality of data and do not assume explicit knowledge of this parameter. Their data structures are akin to the $k$-d tree structure and offer guaranteed reduction in the size of the cells after a bounded number of levels. Such a size reduction is of immense use in vector quantization \cite{vec-quan-rptree} and regression \cite{reg-rptree}. \cite{rptree} presents two such tree structures, each adapting to a different notion of intrinsic dimensionality. Both variants have already found numerous applications in regression \cite{reg-rptree}, spectral clustering \cite{spec-clust-rptree}, face recognition \cite{face-recog-rptree} and image super-resolution \cite{super-res-rptree}.

\subsection{Contributions}
The \RP\ structures are new entrants in a large family of space partitioning data structures such as $k$-d trees \cite{kdtree}, BBD trees \cite{bbdtree}, BAR trees \cite{bartree} and several others (see \cite{moti-kitab} for an overview). The typical guarantees given by these data structures are of the following types :
\begin{enumerate}
	\item \textbf{Space Partitioning Guarantee} : There exists a bound $L(s), s \geq 2$ on the number of levels one has to go down before all descendants of a node of size $\Delta$ are of size $\Delta/s$ or less. The size of a cell is variously defined as the length of the longest side of the cell (for box-shaped cells), radius of the cell, etc. 
	\item \textbf{Bounded Aspect Ratio} : There exists a certain ``roundedness'' to the cells of the tree - this notion is variously defined as the ratio of the length of the longest to the shortest side of the cell (for box-shaped cells), the ratio of the radius of the smallest circumscribing ball of the cell to that of the largest ball that can be inscribed in the cell, etc.
	\item \textbf{Packing Guarantee} : Given a fixed ball $B$ of radius $R$ and a size parameter $r$, there exists a bound on the number of disjoint cells of the tree that are of size greater than $r$ and intersect $B$. Such bounds are usually arrived at by first proving a bound on the aspect ratio for cells of the tree.
\end{enumerate}
These guarantees play a crucial role in algorithms for fast approximate nearest neighbor searches \cite{bbdtree} and clustering \cite{clust-bbdtree}. We present new results for the \RPM\ structure for all these types of guarantees. We first present a bound on the number of levels required for size reduction by any given factor in an \RPM. Our result improves the bound obtainable from results presented in \cite{rptree}. Next, we prove an ``effective'' aspect ratio bound for \RPM. Given the randomized nature of the data structure it is difficult to directly bound the aspect ratios of all the cells. Instead we prove a weaker result that can nevertheless be exploited to give a packing lemma of the kind mentioned above. More specifically, given a ball $B$, we prove an aspect ratio bound for the smallest cell in the \RPM\ that completely contains $B$.

Our final result concerns the \RPm\ data structure. The authors in \cite{rptree} prove that this structure adapts to the \emph{Local Covariance Dimension} of data (see Section~\ref{loc-cov-manifold} for a definition). By showing that low-dimensional manifolds have bounded local covariance dimension, we show its adaptability to the manifold dimension as well. Our result demonstrates the robustness of the notion of manifold dimension - a notion that is able to connect to a geometric notion of dimensionality such as the doubling dimension (proved in \cite{rptree}) as well as a statistical notion such as Local Covariance Dimension (this paper).

\subsection{Organization of the paper}
In Section~\ref{prob-rptree} we present a brief introduction to the \RPM\ data structure and discuss its analysis. In Section~\ref{size-redn} we present our generalized size reduction lemma for the \RPM. In Section~\ref{packing} we give an effective aspect ratio bound for the \RPM\ which we then use to arrive at our packing lemma. In Section~\ref{loc-cov-manifold} we show that the \RPm\ adapts to manifold dimension.

All results cited from other papers are presented as \emph{Facts} in this paper. We will denote by $B(x,r)$, a closed ball of radius $r$ centered at $x$. We will denote by $d$, the intrinsic dimensionality of data and by $D$, the ambient dimensionality (typically $d \ll D$).

\section{The \RPM\ structure}
\label{prob-rptree}
The \RPM\ structure adapts to the doubling dimension of data (see definition below). Since low-dimensional manifolds have low doubling dimension (see \cite{rptree} Theorem 22) hence the structure adapts to manifold dimension as well.
\begin{defn}[taken from \cite{assouad}]
\label{assouad-defn}
The doubling dimension of a set $S \subset \R^D$ is the smallest integer $d$ such that for any ball $B(x,r) \subset \R^D$, the set $B(x,r) \cap S$ can be covered by $2^d$ balls of radius $r/2$.
\end{defn}
The \RPM\ algorithm is presented data imbedded in $\R^D$ having doubling dimension $d$. The algorithm splits data lying in a cell $C$ of radius $\Delta$ by first choosing a random direction $v \in \R^D$, projecting all the data inside $C$ onto that direction, choosing a random value $\delta$ in the range $[-1,1]\cdot 6\Delta/\sqrt{D}$ and then assigning a data point $x$ to the left child if $x\cdot v < \text{median}(\{z\cdot v: z \in C\}) + \delta$ and the right child otherwise. Since it is difficult to get the exact value of the radius of a data set, the algorithm settles for a constant factor approximation to the value by choosing an arbitrary data point $x \in C$ and using the estimate $\tilde{\Delta} = \max(\{\|x-y\| : y \in C\})$.

The following result is proven in \cite{rptree} :
\begin{fct}[Theorem 3 in \cite{rptree}]
\label{main-rptree}
There is a constant $c_1$ with the following property. Suppose an \RPM\ is built using a data set $S \subset \R^D$ . Pick any cell $C$ in the \RPM; suppose that $S \cap C$ has doubling dimension $\leq d$. Then with probability at least $1/2$ (over the randomization in constructing the subtree rooted at $C$), every descendant $C'$ more than $c_1d\log d$ levels below $C$ has radius$(C') \leq$ radius$(C)/2$.
\end{fct}
In Sections~\ref{prob-rptree}, \ref{size-redn} and \ref{packing}, we shall always assume that the data has doubling dimension $d$ and shall not explicitly state this fact again and again. Let us consider extensions of this result to bound the number of levels it takes for the size of all descendants to go down by a factor $s > 2$. Let us analyze the case of $s = 4$. Starting off in a cell $C$ of radius $\Delta$, we are assured of a reduction in size by a factor of $2$ after $c_1d\log d$ levels. Hence all $2^{c_1d\log d}$ nodes at this level have radius $\Delta/2$ or less. Now we expect that after $c_1d\log d$ more levels, the size should go down further by a factor of $2$ thereby giving us our desired result. However, given the large number of nodes at this level and the fact that the success probability in Fact~\ref{main-rptree} is just greater than a constant bounded away from $1$, it is not possible to argue that after $c_1d\log d$ more levels the descendants of all these $2^{c_1d\log d}$ nodes will be of radius $\Delta/4$ or less. It turns out that this can be remedied by utilizing the following extension of the basic size reduction result in \cite{rptree}. We omit the proof of this extension.
\begin{fct}[Extension of Theorem 3 in \cite{rptree}]
\label{main-rptree-extn}
For any $\delta > 0$, with probability at least $1 - \delta$, every descendant $C'$ which is more than $c_1d\log d + \log(1/\delta)$ levels below $C$ has radius$(C') \leq$ radius$(C)/2$.
\end{fct}
This gives us a way to boost the confidence and do the following : go down $L = c_1d\log d + 2$ levels from $C$ to get the the radius of all the $2^{c_1d\log d + 2}$ descendants down to $\Delta/2$ with confidence $1-1/4$. Afterward, go an additional $L' = c_1d\log d + L + 2$ levels from each of these descendants so that for any cell at level $L$, the probability of it having a descendant of radius $> \Delta/4$ after $L'$ levels is less than $\frac{1}{4\cdot2^L}$. Hence conclude with confidence at least $1 - \frac{1}{4} - \frac{1}{4\cdot2^L}\cdot2^L \geq \frac{1}{2}$ that all descendants of $C$ after $2L + c_1d\log d + 2$ have radius $\leq \Delta/4$. This gives a way to prove the following result :
\begin{thm}
\label{weak-gen}
There is a constant $c_2$ with the following property. For any $s \geq 2$, with probability at least $1 - 1/4$, every descendant $C'$ which is more than $c_2\cdot s \cdot d\log d$ levels below $C$ has radius$(C') \leq$ radius$(C)/s$.
\end{thm}
\begin{proof}
Without loss of generality assume that $s$ is a power of $2$. We will prove the result by induction. Fact~\ref{main-rptree-extn} proves the base case for $s = 2$. For the induction step, let $L(s)$ denote the number of levels it takes to reduce the size by a factor of $s$ with high confidence. Then we have
\[
L(s) \leq L(s/2) + c_1d\log d + L(s/2) + 2 = 2L(s/2) + c_1d\log d + 2
\]
Solving the recurrence gives $L(s) = \OO{sd \log d}$
\end{proof}

Notice that the dependence on the factor $s$ is linear in the above result whereas one expects it to be logarithmic. Indeed, typical space partitioning algorithms such as $k$-d trees do give such guarantees. The first result we prove in the next section is a bound on the number of levels that is poly-logarithmic in the size reduction factor $s$.

\section{A generalized size reduction lemma for \RPM}
\label{size-redn}
In this section we prove the following theorem :
\begin{thm}[\textbf{Main}]
\label{strong-gen}
There is a constant $c_3$ with the following property. Suppose an \RPM\ is built using data set $S \subset \R^D$ . Pick any cell $C$ in the \RPM; suppose that $S \cap C$ has doubling dimension $\leq d$. Then for any $s \geq 2$, with probability at least $1 - 1/4$ (over the randomization in constructing the subtree rooted at $C$), for every descendant $C'$ which is more than $c_3 \cdot \log s \cdot d\log sd$ levels below $C$, we have radius$(C') \leq$ radius$(C)/s$.
\end{thm}
Compared to this, data structures such as \cite{bbdtree} give deterministic guarantees for such a reduction in $D\log s$ levels which can be shown to be optimal (see \cite{rptree} for an example). Thus our result is optimal but for a logarithmic factor. Moving on with the proof, let us consider a cell $C$ of radius $\Delta$ in the \RPM\ that contains a dataset $S$ having doubling dimension $\leq d$. Then for any $\epsilon > 0$, a repeated application of Definition~\ref{assouad-defn} shows that the $S$ can be covered using at most $2^{d\log (1/\epsilon)}$ balls of radius $\epsilon\Delta$. We will cover $S \cap C$ using balls of radius $\frac{\Delta}{960s\sqrt{d}}$ so that $\OO{(sd)^d}$ balls would suffice. Now consider all pairs of these balls, the distance between whose centers is $\geq \frac{\Delta}{s} - \frac{\Delta}{960s\sqrt{d}}$. If random splits separate data from all such pairs of balls i.e. for no pair does any cell contain data from both balls of the pair, then each resulting cell would only contain data from pairs whose centers are closer than $\frac{\Delta}{s} - \frac{\Delta}{960s\sqrt{d}}$. Thus the radius of each such cell would be at most $\Delta/s$.
\begin{figure}
	\begin{center}
		\scalebox{0.7}{
			\begin{pspicture}(0,-4.025)(9.620937,4.045)
				\pscircle[linewidth=0.02,dimen=outer](4.58,0.015){3.3}
				\pscircle[linewidth=0.02,dimen=outer](2.74,0.655){0.62}
				\pscircle[linewidth=0.02,dimen=outer](6.49,-0.835){0.65}
				\psline[linewidth=0.02cm,linestyle=dashed](1.56,-3.065)(7.72,2.775)
				\psline[linewidth=0.02cm,linestyle=dashed](0.0,0.955)(8.94,-1.625)
				\psline[linewidth=0.02cm,linestyle=dashed](2.22,3.515)(4.56,-4.005)
				\psline[linewidth=0.02cm,arrows=<->,arrowlength=3,arrowsize=3pt,arrowinset=0.2](4.62,0.015)(5.1,3.235)
				\usefont{T1}{ptm}{m}{n}
				\rput(8.365156,3.065){good split}
				\usefont{T1}{ptm}{m}{n}
				\rput(9.620937,-1.535){bad split}
				\usefont{T1}{ptm}{m}{n}
				\rput(2.7751563,3.865){neutral split}
				\usefont{T1}{ptm}{m}{n}
				\rput(7.121406,-0.215){$B_2$}
				\usefont{T1}{ptm}{m}{n}
				\rput(3.4614062,1.185){$B_1$}
				\usefont{T1}{ptm}{m}{n}
				\rput(5.0914062,1.645){$\Delta$}
			\end{pspicture} 
		}
		\caption{Balls $B_1$ and $B_2$ are of radius $\Delta/s\sqrt{d}$ and their centers are $\Delta/s - \Delta/s\sqrt{d}$ apart.}%
		\label{splits}
	\end{center}
\end{figure}
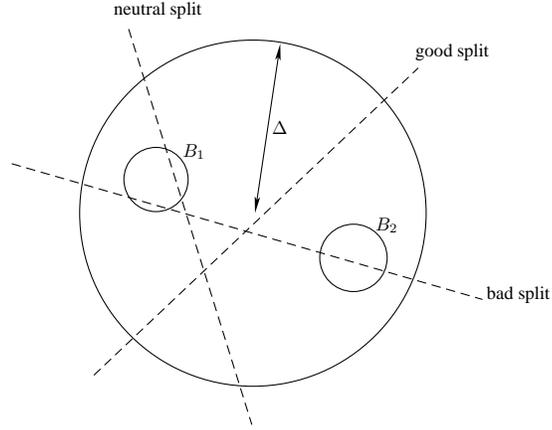

We fix such a pair of balls calling them $B_1$ and $B_2$. A split in the \RPM\ is said to be \emph{good} with respect to this pair if it sends points inside $B_1$ to one child of the cell in the \RPM\ and points inside $B_2$ to the other, \emph{bad} if it sends points from both balls to both children and \emph{neutral} otherwise (See Figure~\ref{splits}). We have the following properties of a random split :
\begin{lem}
\label{ball-split-prob}
Let $B = B(x,\delta)$ be a ball contained inside an \RPM\ cell of radius $\Delta$ that contains a dataset $S$ of doubling dimension $d$. Lets us say that a random split splits this ball if the split separates the data set $S$ into two parts. Then a random split of the cell splits $B$ with probability atmost $\frac{3\delta\sqrt{d}}{\Delta}$.
\end{lem}
\begin{proof}
The \RPM\ splits proceed by randomly projecting the data in a cell onto the real line and then choosing a split point in an interval of length $12\Delta/\sqrt{D}$. It is important to note that the random direction and the split point are chosen independently. Hence, suppose data inside the ball $B$ gets projected onto an interval $\tilde{B}$ of radius $r$, then the probability of it getting split is atmost $r\sqrt{D}/6\Delta$ since the split point is chosen randomly in an interval of length $12\Delta/\sqrt{D}$ independently of the projection. Let $R_B$ be the random variable that gives the radius of the interval $\tilde{B}$. Hence the probability of $B$ getting split is the following
\begin{eqnarray*}
\frac{\sqrt{D}}{6\Delta}\int\limits_0^\infty{r\Pr{R_B=r}dr} &=& \frac{\sqrt{D}}{6\Delta}\int\limits_0^\infty{\int\limits_0^r\Pr{R_B=r}dtdr} = \frac{\sqrt{D}}{6\Delta}\int\limits_0^\infty{\int\limits_t^\infty\Pr{R_B=r}drdt}\\
 &=& \frac{\sqrt{D}}{6\Delta}\int\limits_0^\infty{Pr[R_B \geq t]dt}
\end{eqnarray*}
We have the following result from \cite{rptree}
\begin{fct}[Lemma 6 of \cite{rptree}]
\label{ball-proj}
$\Pr{R_B \geq \frac{4\delta}{\sqrt{D}}\sqrt{2\left(d + \ln\frac{2}{\eta}\right)}} \leq \eta$
\end{fct}
Fix the value $l = \frac{4\delta}{\sqrt{D}}\sqrt{2\left(d + \ln 2\right)}$. Using the fact that for any $t$, $Pr[R_B \geq t] \leq 1$ and making the change of variables $t = \frac{4\delta}{\sqrt{D}}\sqrt{2\left(d + \ln\frac{2}{\eta}\right)}$ we get
\[
\int\limits_0^\infty{Pr[R_B \geq t]dt} = \int\limits_0^l{Pr[R_B \geq t]dt} + \int\limits_l^\infty{Pr[R_B \geq t]dt}  \leq \int\limits_0^l{1 dt} + \int\limits_1^0{\eta dt(\eta)}
\]
Simplifying the above expression, we get the split probability to be atmost
\[
\frac{2\delta}{3\Delta}\left[ \sqrt{2\left(d + \ln 2\right)} + \int\limits_0^1{\frac{d\eta}{\sqrt{2\left(d + \ln\frac{2}{\eta}\right)}}}\right] = \frac{2\delta}{3\Delta}\left[ \sqrt{2\left(d + \ln 2\right)} + 2\sqrt{2}e^d \int\limits_{\sqrt{\ln 2 + d}}^\infty{e^{-x^2}dx}\right]
\]
Now $\int\limits_a^\infty{e^{-x^2}dx} = \frac{1}{2}\left[\int\limits_{-\infty}^\infty{e^{-x^2}dx} - \int\limits_{-a}^a{e^{-x^2}dx}\right] \leq \frac{\sqrt{\pi}}{2}\left[1 - \sqrt{1-e^{-a^2}} \right] \leq \frac{\sqrt{\pi}}{2}e^{-a^2}$ since $1- \sqrt{1-x} < x$ for $0 < x < 1$. Using $d \geq 1$ , we get the probability of the ball $B$ getting split to be atmost $\frac{2\delta}{3\Delta}\left[ \sqrt{2\left(d + \ln 2\right)} + \sqrt{\frac{\pi}{2}} \right] \leq \frac{3\delta\sqrt{d}}{\Delta}$.
\end{proof}

\begin{lem}
\label{good-split-prob}
Let $B_1$ and $B_2$ be a pair of balls as described above contained in the cell $C$ that contains data of doubling dimension $d$. Then a random split of the cell is a good split with respect to this pair with probability at least $\frac{1}{56s}$.
\end{lem}
\begin{proof}
The techniques used in the proof of this lemma are the same as those used to prove a similar result in \cite{rptree}. We are giving a proof sketch here for completeness. We use the following two results from \cite{rptree}
\begin{fct}[Lemma 5 of \cite{rptree}]
\label{norm}
Fix any $x \in \R^D$. Pick a random vector $U \sim \N{0}{(1/D)I_D}$. Then for any $\alpha, \beta > 0$ :
\begin{enumerate}
	\item \label{norm-a} $\Pr{|U \cdot x| \leq \alpha \cdot\frac{\|x\|}{\sqrt{D}}}\leq \sqrt{\frac{2}{\pi}\alpha}$,
	\item \label{norm-b} $\Pr{|U \cdot x| \geq \beta \cdot\frac{\|x\|}{\sqrt{D}}}\leq \frac{2}{\beta}e^{-\beta^2/2}$.
\end{enumerate}
\end{fct}
\begin{fct}[Corollary 8 of \cite{rptree}]
\label{median}
Suppose $S \subset \R^D$ lies within ball $B(x,\Delta)$. Pick any $0 < \delta < 2/e^2$. Let this set be projected randomly onto the real line. Let us denote by $\tilde{x}$, the projection of $x$ by $\tilde{S}$, the projection of the set $S$. Then with probability atleast $1-\delta$ over the choice of random projection onto $\R$, $\left|\text{median}\{\tilde{S}\} - \tilde{x}\right| \leq \frac{\Delta}{\sqrt{D}}\cdot\sqrt{2\ln\frac{2}{\delta}}$.
\end{fct}
Projections of points, sets etc. are denoted with a tilde (\ $\tilde{}$\ ) sign. Applying Fact~\ref{ball-proj} with $\eta = \frac{2}{e^{31}}$, we get that with probability $> 1 - \frac{2}{e^{31}}$, the ball $B_1$ gets projected to an interval of length atmost $\frac{\Delta}{30s\sqrt{D}}$ centered at $\tilde{x_1}$. The same holds for $B_2$. Applying Fact~\ref{norm}\ref{norm-a} with $\alpha = \frac{384}{959}$ gives us $\left|\tilde{x_1} - \tilde{x_2} \right| \geq \frac{\Delta}{2s\sqrt{D}}$ with probability $1 - \frac{1536}{4795}$. Furthermore, an application of Fact~\ref{norm}\ref{norm-b} with $\beta = \sqrt{2\ln 40}$ shows that with probability atleast $1 - \frac{1}{54}$, $\left|\tilde{x_1} - \tilde{x}\right| \leq \frac{3\Delta}{\sqrt{D}}$. The same holds true for $\tilde{x_2}$ as well. Finally an application of Fact~\ref{median} with $\delta = \frac{1}{20}$ shows that the median of the projected set $\tilde{S}$ will lie within a distance $\frac{3\Delta}{\sqrt{D}}$ of $\tilde{x}$ (i.e. the projection of the center of the cell) with probability atleast $1 - \frac{1}{20}$.

Simple calculations show that the preceding guarantees imply that with probability atleast $\frac{1}{2}$ over the choice of random projections, the projections of both the balls will lie within the interval from which a split point would be chosen. Further more there would be a gap of atleast $\frac{\Delta}{2s\sqrt{D}} - 2\frac{\Delta}{30s\sqrt{D}}$ between the projections of the two balls. Hence, given that these good events take place, with probability atleast $\frac{\sqrt{D}}{12\Delta}\left(\frac{\Delta}{2s\sqrt{D}} - 2\frac{\Delta}{30s\sqrt{D}}\right)$ over the choice of the split point, the balls will get cleanly separated. Note that this uses independence of the choice of projection and the choice of the split point. Thus the probability of a good split is atleast $\frac{1}{56s}$.
\end{proof}

\begin{lem}
\label{bad-split-prob}
Let $B_1$ and $B_2$ be a pair of balls as described above contained in the cell $C$ that contains data of doubling dimension $d$. Then a random split of the cell is a bad split with respect to this pair with probability at most $\frac{1}{320s}$.
\end{lem}
\begin{proof}
The proof of a similar result in \cite{rptree} uses a conditional probability argument. However the technique does not work here since we require a bound that is inversely proportional to $s$. We instead make a simple observation that the probability of a bad split is upper bounded by the probability that one of the balls is split since for any two events $A$ and $B$, $\Pr{A \cap B} \leq \min\{\Pr{A},\Pr{B}\}$. The result then follows from an application of Lemma~\ref{ball-split-prob}.
\end{proof}
We are now in a position to prove Theorem~\ref{strong-gen}. What we will prove is that starting with a pair of balls in a cell $C$, the probability that some cell $k$ levels below has data from both the balls is exponentially small in $k$. Thus, after going enough number of levels we can take a union bound over all pairs of balls whose centers are well separated (which are $\OO{(sd)^{2d}}$ in number) and conclude the proof.
\begin{proof}(of Theorem~\ref{strong-gen})
Consider a cell $C$ of radius $\Delta$ in the \RPM\ and fix a pair of balls contained inside $C$ with radii $\Delta/960s\sqrt{d}$ and centers separated by at least $\Delta/s - \Delta/960s\sqrt{d}$. Let $p^i_j$ denote the probability that a cell $i$ levels below $C$ has a descendant $j$ levels below itself that contains data points from both the balls. Then the following holds :
\begin{lem}
\label{go-down}
$p^0_k \leq \left(1 - \frac{1}{68s} \right)^lp^l_{k-l}$.
\end{lem}
\begin{proof}
We have the following expression for $p^0_k$ :
\begin{eqnarray*}
p_k^0 &\leq& \Pr{\text{split at level $0$ is a good split}} \cdot 0 +\\
      &&     \Pr{\text{split at level $0$ is a bad split}} \cdot 2p^1_{k-1} +\\
      &&     \Pr{\text{split at level $0$ is a neutral split}} \cdot p^1_{k-1}\\
      &\leq&    \frac{1}{320s} \cdot 2p^1_{k-1} + \left(1 - \frac{1}{320s} - \frac{1}{56s} \right) \cdot p^1_{k-1}\\
      &=&    \left(1 + \frac{1}{320s} - \frac{1}{56s} \right) \cdot p^1_{k-1}\\
      &=&    \left(1 - \frac{1}{68s} \right)p^1_{k-1}\\
      &\leq&    \left(1 - \frac{1}{68s} \right)^2p^2_{k-2}\qquad \left(\text{Similarly } p^1_{k-1} \leq \left(1 - \frac{1}{68s} \right)p^2_{k-2} \right)\\
      &&     \vdots\\
      &\leq&    \left(1 - \frac{1}{68s} \right)^lp^l_{k-l}
\end{eqnarray*}
\vspace*{-20pt}\[\qedhere\]
\end{proof}
Note that this gives us $p^0_k \leq \left(1 - \frac{1}{68s} \right)^k$ as a corollary. However using this result would require us to go down $k = \Omega(sd\log(sd))$ levels before $p^0_k = \frac{1}{\Omega\left( (sd)^{2d} \right)}$ which results in a bound that is worse (by a factor logarithmic in $s$) than the one given by Theorem~\ref{weak-gen}. This can be attributed to the small probability of a good split for a tiny pair of balls in large cells. However, here we are completely neglecting the fact that as we go down the levels, the radii of cells go down as well and good splits become more frequent.

Indeed setting $s = 2$ in Theorems~\ref{good-split-prob} and \ref{bad-split-prob} tells us that if the pair of balls were to be contained in a cell of radius $\frac{\Delta}{s/2}$ then the good and bad split probabilities are $\frac{1}{112}$ and $\frac{1}{640}$ respectively. This paves way for an inductive argument : assume that with probability $> 1 - 1/4$, in $L(s)$ levels, the size of all descendants go down by a factor $s$. Denote by $p_g^l$ the probability of a good split in a cell at depth $l$ and by $p_b^l$ the corresponding probability of a bad split. Set $l^\ast = L(s/2)$ and let $E$ be the event that the radius of every cell at level $l^\ast$ is less than $\frac{\Delta}{s/2}$. Let $C'$ represent a cell at depth $l^\ast$. Then,
\begin{eqnarray*}
p_g^{l^\ast} &\geq& \Pr{\text{good split in $C' | E$}}\cdot\Pr{E} \geq \frac{1}{112}\cdot\left(1 - \frac{1}{4}\right) \geq \frac{1}{150}\\
p_b^{l^\ast} &=& \Pr{\text{bad split in $C' | E$}}\cdot\Pr{E} + \Pr{\text{bad split in $C' | \neg E$}}\cdot\Pr{\neg E}\\
 &\leq& \frac{1}{640}\cdot 1 +  \frac{1}{640}\cdot\frac{1}{4} \leq \frac{1}{512}
\end{eqnarray*}
Notice that now, for any $m > 0$, we have $p^{l^\ast}_m \leq \left(1 - \frac{1}{213}\right)^m$. Thus, for some constant $c_4$, setting $k = l^\ast + c_4 d\log (sd)$ and applying Lemma~\ref{go-down} gives us $p^0_k \leq \left(1 - \frac{1}{68s} \right)^{l^\ast}\left(1 - \frac{1}{213} \right)^{c_4d\log(sd)} \leq \frac{1}{4(sd)^{2d}}$. Thus we have
\[
L(s) \leq L(s/2) + c_4 d\log(sd)
\]
which gives us the desired result on solving the recurrence i.e. $L(s) = \OO{d\log s\log sd}$.
\end{proof}

\section{A packing lemma for \RPM}
\label{packing}
In this section we prove a probabilistic packing lemma for \RPM. A formal statement of the result follows :
\begin{thm}[\textbf{Main}]
\label{packing-thm}
Given any fixed ball $B(x,R) \subset \R^D$, with probability greater than $1/2$ (where the randomization is over the construction of the \RPM), the number of disjoint \RPM\ cells of radius greater than $r$ that intersect $B$ is at most $\left( \frac{R}{r} \right)^{\OO{d\log d\log(dR/r)}}$.
\end{thm}
Data structures such as BBD-trees give a bound of the form $\OO{\frac{R}{r}}^D$ which behaves like $\left(\frac{R}{r}\right)^{\OO{1}}$ for fixed $D$. In comparison, our result behaves like  $\left(\frac{R}{r}\right)^{\OO{\log \frac{R}{r}}}$ for fixed $d$. We will prove the result in two steps : first of all we will show that with high probability, the ball $B$ will be completely inscribed in an \RPM\ cell $C$ of radius no more than $\OO{Rd\sqrt{d}\log d}$. Thus the number of disjoint cells of radius at least $r$ that intersect this ball is bounded by the number of descendants of $C$ with this radius. To bound this number we then invoke Theorem~\ref{strong-gen} and conclude the proof.

\subsection{An effective aspect ratio bound for \RPM\ cells}
\label{aspect-ratio}
In this section we prove an upper bound on the radius of the smallest \RPM\ cell that completely contains a given ball $B$ of radius $R$. Note that this effectively bounds the aspect ratio of this cell.
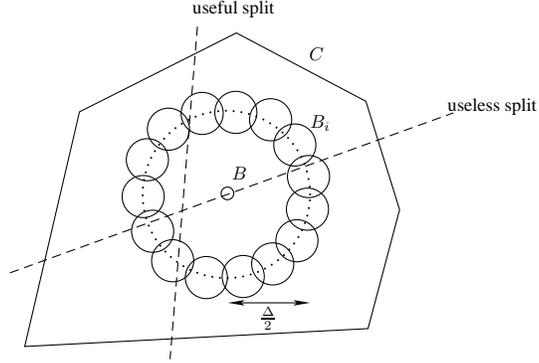
\begin{figure}
	\begin{center}
		\scalebox{0.7}{
			\begin{pspicture}(0,-4.265)(9.800625,4.285)
				\pspolygon[linewidth=0.02](4.3,2.955)(1.32,1.455)(0.28,-3.005)(6.8,-2.665)(7.4,-0.405)(6.76,1.655)
				\pscircle[linewidth=0.02,dimen=outer](4.13,-0.095){0.13}
				\pscircle[linewidth=0.02,dimen=outer](4.43,-1.675){0.41}
				\pscircle[linewidth=0.02,dimen=outer](3.09,-1.375){0.41}
				\pscircle[linewidth=0.02,dimen=outer](3.73,-1.695){0.41}
				\pscircle[linewidth=0.02,dimen=outer](2.71,-0.815){0.41}
				\pscircle[linewidth=0.02,dimen=outer](2.53,-0.155){0.41}
				\pscircle[linewidth=0.02,dimen=outer](2.61,0.545){0.41}
				\pscircle[linewidth=0.02,dimen=outer](3.01,1.125){0.41}
				\pscircle[linewidth=0.02,dimen=outer](4.99,-1.435){0.41}
				\pscircle[linewidth=0.02,dimen=outer](3.65,1.425){0.41}
				\pscircle[linewidth=0.02,dimen=outer](4.29,1.445){0.41}
				\pscircle[linewidth=0.02,dimen=outer](4.95,1.305){0.41}
				\pscircle[linewidth=0.02,dimen=outer](5.41,0.825){0.41}
				\pscircle[linewidth=0.02,dimen=outer](5.67,0.225){0.41}
				\pscircle[linewidth=0.02,dimen=outer](5.45,-0.995){0.41}
				\pscircle[linewidth=0.02,dimen=outer](5.65,-0.395){0.41}
				\pscircle[linewidth=0.04,dimen=outer,linestyle=dotted](4.11,-0.115){1.61}
				\psline[linewidth=0.02cm,arrows=<->,arrowlength=3,arrowsize=3pt,arrowinset=0.2]
				(4.14,-2.175)(5.7,-2.175)
				\usefont{T1}{ptm}{m}{n}
				\rput(4.9014063,-2.465){$\frac{\Delta}{2}$}
				\psline[linewidth=0.02cm,linestyle=dashed](3.56,3.075)(3.04,-3.245)
				\psline[linewidth=0.02cm,linestyle=dashed](0.0,-1.605)(8.42,1.435)
				\usefont{T1}{ptm}{m}{n}
				\rput(4.245,3.405){useful split}
				\usefont{T1}{ptm}{m}{n}
				\rput(9.160781,1.565){useless split}
				\usefont{T1}{ptm}{m}{n}
				\rput(5.8914062,1.245){$B_i$}
				\usefont{T1}{ptm}{m}{n}
				\rput(4.3614063,0.285){$B$}
				\usefont{T1}{ptm}{m}{n}
				\rput(5.8214064,2.565){$C$}
			\end{pspicture} 
		}
		\caption{Balls $B_i$ are of radius $\Delta/512\sqrt{d}$ and their centers are $\Delta/2$ far from the center of $B$.}
		\label{aspect}
	\end{center}
\end{figure}
Consider any cell $C$ of radius $\Delta$ that contains $B$. We proceed with the proof by first showing that the probability that $B$ will be split before it lands up in a cell of radius $\Delta/2$ is at most a quantity inversely proportional to $\Delta$. Note that we are not interested in all descendants of $C$ - only  the ones ones that contain $B$. That is why we argue differently here. We consider balls of radius $\Delta/512\sqrt{d}$ surrounding $B$ at a distance of $\Delta/2$ (see Figure~\ref{aspect}). These balls are made to cover the annulus centered at $B$ of mean radius $\Delta/2$ and thickness $\Delta/512\sqrt{d}$ -- clearly $d^{\OO{d}}$ balls suffice. Without loss of generality assume that the centers of all these balls lie in $C$.

Notice that if $B$ gets separated from all these balls without getting split in the process then it will lie in a cell of radius $< \Delta/2$. Fix a $B_i$ and call a random split of the \RPM\ \emph{useful} if it separates $B$ from $B_i$ and \emph{useless} if it splits $B$. Using a proof technique similar to that used in Lemma~\ref{good-split-prob} we can show that the probability of a useful split is at least $\frac{1}{192}$ whereas Lemma~\ref{ball-split-prob} tells us that the probability of a useless split is at most $\frac{3R\sqrt{d}}{\Delta}$.

\begin{lem}
\label{ball-go-down}
There exists a constant $c_5$ such that the probability of a ball of radius $R$ in a cell of radius $\Delta$ getting split before it lands up in a cell of radius $\Delta/2$ is at most $\frac{c_5Rd\sqrt{d}\log d}{\Delta}$.
\end{lem}
\begin{proof}
The only bad event for us is the one in which $B$ gets split before it gets separated from all the $B_j$'s. Call this event $E$. Also, denote by $E[i]$ the bad event that $B$ gets split for the first time in the $i^{\text{th}}$ split and the preceding $i-1$ splits are incapable of separating $B$ from all the $B_j$'s. Thus $\Pr{E} \leq \sum\limits_{i > 0}{\Pr{E[i]}}$. Since any given split is a useful split (i.e. separates $B$ from a fixed $B_j$) with probability $> \frac{1}{192}$, the probability that $i-1$ splits will fail to separate all $B_j$s from the $B$ (while not splitting $B$) is at most $\min\left\{1, \left(1 - \frac{1}{192}\right)^{i-1}\cdot N \right\}$ where $N = d^{\OO{d}}$ is the number of balls $B_j$. Since all splits in an \RPM\ are independent of each other, we have $\Pr{E[i]} \leq \min\left\{1, \left(1 - \frac{1}{192}\right)^{i-1}\cdot N \right\}\cdot\frac{3R\sqrt{d}}{\Delta}$. Let $k$ be such that $\left(1 - \frac{1}{192}\right)^{k-1} \leq \frac{1}{4N}$. Clearly $k = \OO{d\log d}$ suffices. Thus we have
\[
\Pr{E} \leq \frac{3R\sqrt{d}}{\Delta}\sum\limits_{i > 0}{\min\left\{1, \left(1 - \frac{1}{192}\right)^{i-1}\cdot N \right\}} \leq \frac{3R\sqrt{d}}{\Delta}\left(\sum\limits_{i=1}^k{1} + \sum\limits_{i=1}^\infty{\frac{1}{4}\left(1 - \frac{1}{192}\right)^i}\right)
\]
which gives us $\Pr{E} = \OO{\frac{Rd\sqrt{d}\log d}{\Delta}}$ since the second summation is just a constant.
\end{proof}
We now state our result on the ``effective'' bound on aspect ratios of \RPM\ cells.
\begin{thm}
\label{aspect-ratio-proof}
There exists a constant $c_6$ such that with probability $> 1 - 1/4$, a given (fixed) ball $B$ of radius $R$ will be completely inscribed in an \RPM\ cell $C$ of radius no more than $c_6\cdot Rd\sqrt{d}\log d$.
\end{thm}
\begin{proof}
Let $\Delta^\ast = 4c_5Rd\sqrt{d}\log d$ and $\Delta_{\max}$ be the radius of the entire dataset. Denote by $F[i]$ the event that $B$ ends up unsplit in a cell of radius $\frac{\Delta_{\max}}{2^i}$. The event we are interested in is $F[m]$ for $m = \log\frac{\Delta_{\max}}{\Delta^\ast}$. Note that $\Pr{F[m]|F[m-1]}$ is exactly $\Pr{E}$ where $E$ is the event described in Lemma~\ref{ball-go-down} for appropriately set value of radius $\Delta$. Also $\Pr{F[m]|\neg F[m-1]} = 0$. Thus we have
\begin{eqnarray*}
\Pr{F[m]} &=& \prod\limits_{i=0}^{m-1}{\Pr{\left. F[i+1] \right| F[i]}} = \prod\limits_{i=0}^{m-1}{\left(1 - \frac{c_5Rd\sqrt{d}\log d}{\Delta_{\max}/2^i}\right)} \geq 1 - \sum\limits_{i=0}^{m-1}{\frac{c_5Rd\sqrt{d}\log d}{\Delta_{\max}/2^i}}\\
          &=& 1 - \sum\limits_{i=0}^{m-1}{\frac{c_5Rd\sqrt{d}\log d}{2^{m-i}\Delta^\ast}} = 1 - \frac{1}{4} \sum\limits_{i=0}^{m-1}{\frac{1}{2^{m-i}}} \geq 1 - \frac{1}{4}
\end{eqnarray*}
Setting $c_6 = 4c_5$ gives us the desired result.
\end{proof}
\begin{proof}(of Theorem~\ref{packing-thm})
Given a ball $B$ of radius $R$, Theorem~\ref{aspect-ratio-proof} shows that with probability at least $3/4$, $B$ will lie in a cell $C$ of radius at most $R' = \OO{Rd\sqrt{d}\log d}$. Hence all cells of radius atleast $r$ that intersect this ball must be either descendants or ancestors of $C$. Since we want an upper bound on the largest number of such disjoint cells, it suffices to count the number of descendants of $C$ of radius no less than $r$. We know from Theorem~\ref{strong-gen} that with probability at least $3/4$ in $\log(R'/r)d\log(dR'/r)$ levels the radius of all cells must go below $r$. The result follows by observing that the \RPM\ is a binary tree and hence the number of children can be at most $2^{\log(R'/r)d\log(dR'/r)}$. The success probability is at least $(3/4)^2 > 1/2$.
\end{proof}

\section{Local covariance dimension of a smooth manifold}
\label{loc-cov-manifold}
The second variant of \RP, namely \RPm, adapts to the local covariance dimension (see definition below) of data. We do not go into the details of the guarantees presented in \cite{rptree} due to lack of space. Informally, the guarantee is of the following kind : given data that has small local covariance dimension, on expectation, a data point in a cell of radius $r$ in the \RPm\ will be contained in a cell of radius $c_7\cdot r$ in the next level for some constant $c_7 < 1$. The randomization is over the construction of \RPm\ as well as choice of the data point. This gives per-level improvement albeit in expectation whereas \RPM\ gives improvement in the worst case but after a certain number of levels.

We will prove that a $d$-dimensional Riemannian submanifold $\M$ of $\R^D$ has bounded local covariance dimension thus proving that \RPm\ adapts to manifold dimension as well.
\begin{defn}
A set $S \subset \R^D$ has local covariance dimension $(d,\epsilon,r)$ if there exists an isometry $M$ of $\R^D$ under which the set $S$ when restricted to any ball of radius $r$ has a covariance matrix for which some $d$ diagonal elements contribute a $(1 - \epsilon)$ fraction of its trace.
\end{defn}
This is a more general definition than the one presented in \cite{rptree} which expects the top $d$ eigenvalues of the covariance matrix to account for a $(1 - \epsilon)$ fraction of its trace. However, all that \cite{rptree} requires for the guarantees of \RPm\ to hold is that there exist $d$ orthonormal directions such that a $(1 - \epsilon)$ fraction of the energy of the dataset i.e. $\sum_{x\in S}{\|x - mean(S)\|^2}$ is contained in those $d$ dimensions. This is trivially true when $\M$ is a $d$-dimensional affine set. However we also expect that for small neighborhoods on smooth manifolds, most of the energy would be concentrated in the tangent plane at a point in that neighborhood (see Figure~\ref{manifold}). Indeed, we can show the following :
\begin{thm}[\textbf{Main}]
\label{loc-cov}
Given a data set $S \subset \M$ where $\M$ is a $d$-dimensional Riemannian manifold with condition number $\tau$, then for any $\epsilon \leq \frac{1}{4}$, $S$ has local covariance dimension $\left(d,\epsilon,\frac{\sqrt{\epsilon}\tau}{3}\right)$.
\end{thm}
For manifolds, the local curvature decides how small a neighborhood should one take in order to expect a sense of ``flatness'' in the non-linear surface. This is quantified using the \emph{Condition Number} $\tau$ of $\M$ (introduced in \cite{homology}) which restricts the amount by which the manifold can curve locally. The condition number is related to more prevalent notions of local curvature such as the second fundamental form \cite{diff-geom} in that the inverse of the condition number upper bounds the norm of the second fundamental form \cite{homology}. Informally, if we restrict ourselves to regions of the manifold of radius $\tau$ or less, then we get the requisite flatness properties. \cite{homology} formalizes this as follows. For any hyperplane $T \subset \R^D$ and a vector $v \in \R^d$, let $\parr{v}(T)$ denote the projection of $v$ onto $T$.
\begin{figure}
	\begin{center}
		\scalebox{0.7}{
			\begin{pspicture}(0,-3.0779946)(13.3128,3.0779948)		\psbezier[linewidth=0.04](2.4728,-3.0579948)(2.7928,-1.5979948)(4.078631,-0.8658259)(4.6128,-1.6379948)(5.146969,-2.4101636)(6.8334246,-2.893331)(7.5128,-2.1379948)(8.192176,-1.3826586)(6.4528003,-0.91799474)(6.8128,0.16200523)(7.1727996,1.2420052)(7.3128,1.8020053)(6.7528,1.5820053)(6.1928,1.3620054)(5.570686,0.021776741)(4.6728,0.46200523)(3.774914,0.9022337)(5.253406,2.6660156)(4.2728,2.8620052)(3.292194,3.0579948)(3.1831093,-0.59687483)(2.1928,-0.45799476)(1.2024907,-0.31911469)(2.6324172,1.4543375)(1.6328001,1.4820052)(0.633183,1.5096729)(0.0,0.08170036)(0.052800056,-0.83799475)(0.10560011,-1.7576898)(1.1128,-2.4779947)(2.0328,-1.6379948)
				\pscircle[linewidth=0.04,linestyle=dotted,dotsep=0.10583334cm,dimen=outer](6.8728,-0.2979948){0.6}
				\pscircle[linewidth=0.04,linestyle=dotted,dotsep=0.10583334cm,dimen=outer](11.7028,-0.3079947){1.61}
				\psbezier[linewidth=0.04](12.38,-1.7420053)(11.7,-1.3220053)(10.72,0.09799469)(11.363023,1.2579947)
				\psline[linewidth=0.02cm,linestyle=dashed,dash=0.17638889cm 0.10583334cm](6.78,0.31799468)(11.34,1.2779946)
				\psline[linewidth=0.02cm,linestyle=dashed,dash=0.17638889cm 0.10583334cm](6.74,-0.88200533)(11.42,-1.8820053)
				\psline[linewidth=0.02cm](10.8,1.8179947)(11.64,-2.2620053)
				\usefont{T1}{ptm}{m}{n}
				\rput(11.412812,0.18799467){$p$}
				\usefont{T1}{ptm}{m}{n}
				\rput(11.412812,1.9479947){$T_p(\M)$}
				\usefont{T1}{ptm}{m}{n}
				\rput(5.0928125,2.1279948){$\M$}
				\pscircle*[linewidth=0.04,dimen=outer](11.17,0.04799469){0.07}
			\end{pspicture} 
		}
		\caption{Locally, almost all the energy of the data is concentrated in the tangent plane.}
		\label{manifold}
	\end{center}
\end{figure}
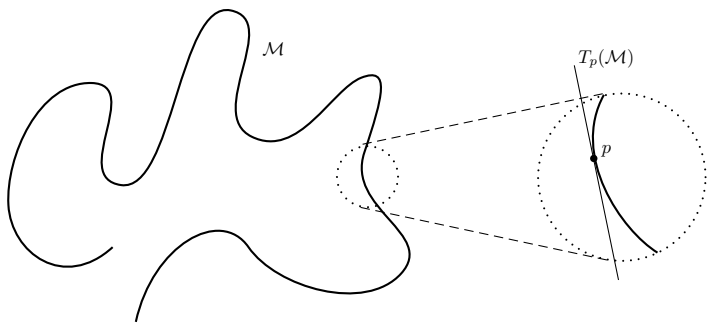
\begin{fct}[Implicit in Lemma 5.3 of \cite{homology}]
\label{proj-not-small}
Suppose $\M$ is a Riemannian manifold with condition number $\tau$. For any $p \in \M$ and $r \leq \sqrt{\epsilon}\tau, \epsilon \leq \frac{1}{4}$, let $\M' = B(p,r) \cap \M$. Let $T = T_p(\M)$ be the tangent space at $p$. Then for any $x,y \in \M'$, $\|\parr{x}(T) - \parr{y}(T)\|^2 \geq (1 - \epsilon)\|x - y\|^2$.
\end{fct}

This already seems to give us what we want - a large fraction of the length between any two points on the manifold lies in the tangent plane - i.e. in $d$ dimensions. However in our case we have to show that for some $d$-dimensional plane $P$, $\sum_{x\in S}{\|\parr{(x - \mu)}(P)\|^2} > (1- \epsilon)\sum_{x\in S}{\|x - \mu\|^2}$ where $\mu = mean(S)$. The problem is that we cannot apply Fact~\ref{proj-not-small} since there is no surety that the mean will lie on the manifold itself. However it turns out that certain points on the manifold can act as ``proxies'' for the mean and provide a workaround to the problem.
\begin{proof}(of Theorem~\ref{loc-cov})
Suppose $\M' = B(x_0,r) \cap \M$ for $r = \frac{\sqrt{\epsilon}\tau}{3}$ and we are given data points $S = \{x_1,\ldots x_n\} \subset \M'$. Let $q = \mathop{\arg\min}\limits_{x \in \M}{\|\mu - x\|}$ be the closest point on the manifold to the mean. The smoothness properties of $\M$ tell us that the vector $(\mu - q)$ is perpendicular to $T_q(\M)$, the $d$-dimensional tangent space at $q$ (in fact any point $q$ at which the function $g: x \in \M \longmapsto \|x - \mu\|$ attains a local extrema would also have the same property). This has interesting consequences - let $f$ be the projection map onto $T_q(\M)$ i.e. $f(v) = \parr{v}(T_q(\M))$.

Then $f(\mu - q) = 0$ since $(\mu - q)\ \bot\ T_q(\M)$. This implies that for any vector $v \in \R^D$, $f(v - \mu) = f(v - q) + f(q - \mu) =  f(v - q) = f(v) - f(q)$ since $f$ is a linear map. We now note that $\underset{i}\min \|\mu - x_i\| \leq r$. If this were not true then we would have $\sum\limits_i\|\mu - x_i\| > nr^2$ whereas we know that $\sum\limits_i\|\mu - x_i\| \leq \sum\limits_i\|x_0 - x_i\| \leq nr^2$ since for any random variable $X \in \R^D$ and fixed $v \in \R^D$, we have $\E{\|X - v\|^2} \geq \E{\|X - \E{X}\|^2}$. Since $\|\mu - x_i\| \leq r$ for some $x_i \in \M$, we know, by definition of $q$, that $\|\mu - q\| \leq r$ as well.

We also have $\|\mu - x_0\| \leq r$ (since the convex hull of the points is contained in the ball $B$ and the mean, being a convex combination of the points, is contained in the hull) and $\|x_i - x_0\| \leq r$ for all points $x_i$. Hence we have for any point $x_i$, $\|x_i - q\| \leq \|x_i - x_0\| + \|x_0 - \mu\| + \|\mu - q\| \leq 3r$ and conclude that $S \subset B(q,3r) \cap \M = B(q,\sqrt{\epsilon}\tau) \cap \M$ which means we can apply Fact~\ref{proj-not-small} between the vectors $x_i$ and $q$.

Let $T = T_q(\M)$ and $q$ as chosen above. We have
\begin{eqnarray*}
\sum\limits_{x \in S}{\|\parr{(x-\mu)}(T)\|^2} &=& \sum\limits_{x \in S}{\|f(x - \mu)\|^2} = \sum\limits_{x \in S}{\|f(x - q)\|^2} = \sum\limits_{x \in S}{\|f(x) - f(q)\|^2}\\
&\geq& \sum\limits_{x \in S}{(1 - \epsilon)\|x - q\|^2} \geq (1 - \epsilon)\sum\limits_{x \in S}{\|x - \mu\|^2}
\end{eqnarray*}
where the last inequality again uses the fact that for a random variable $X \in \R^D$ and fixed $v \in \R^D$, $\E{\|X - v\|^2} \geq \E{\|X - \E{X}\|^2}$.
\end{proof}

\section{Conclusion}
In this paper we considered the two random projection trees proposed in \cite{rptree}. For the \RPM\ data structure, we provided an improved bound (Theorem~\ref{strong-gen}) on the number of levels required to decrease the size of the tree cells by any factor $s \geq 2$. However the bound we proved is poly-logarithmic in $s$. It would be nice if this can be brought down to logarithmic since it would directly improve the packing lemma (Theorem~\ref{packing-thm}) as well. More specifically the packing bound would become $\left(\frac{R}{r}\right)^{\OO{1}}$instead of $\left(\frac{R}{r}\right)^{\OO{\log \frac{R}{r}}}$ for fixed $d$.

As far as dependence on $d$ is concerned, there is room for improvement in the packing lemma. We have shown that the smallest cell in the \RPM\ that completely contains a fixed ball $B$ of radius $R$ has an aspect ratio no more than $\OO{d\sqrt{d}\log d}$ since it has a ball of radius $R$ inscribed in it and can be circumscribed by a ball of radius no more than $\OO{Rd\sqrt{d}\log d}$. Any improvement in the aspect ratio of the smallest cell that contains a given ball will also directly improve the packing lemma.

Moving on to our results for the \RPm, we demonstrated that it adapts to manifold dimension as well. However the constants involved in our guarantee are pessimistic. For instance, the radius parameter in the local covariance dimension is given as $\frac{\sqrt{\epsilon}\tau}{3}$ - this can be improved to $\frac{\sqrt{\epsilon}\tau}{2}$ if one can show that there will always exists a point $q \in B(x_0,r) \cap \M$ at which the function $g: x \in \M \longmapsto \|x - \mu\|$ attains a local extrema.

We conclude with a word on the applications of our results. As we already mentioned, packing lemmas and size reduction guarantees for arbitrary factors are typically used in applications for nearest neighbor searching and clustering. However, these applications (viz \cite{bbdtree}, \cite{clust-bbdtree}) also require that the tree have bounded depth. The \RPM\ is a pure space partitioning data structure that can be coerced by an adversarial placement of points into being a primarily left-deep or right-deep tree having depth $\Omega(n)$ where $n$ is the number of data points.

Existing data structures such as BBD Trees remedy this by alternating space partitioning splits with data partitioning splits. Thus every alternate split is forced to send at most a constant fraction of the points into any of the children thus ensuring a depth that is logarithmic in the number of data points. \cite{reg-rptree} also uses a similar technique to bound the depth of the version of \RPM\ used in that paper. However it remains to be seen if the same trick can be used to bound the depth of \RPM\ while maintaining the packing guarantees because although such ``space partitioning'' splits do not seem to hinder Theorem~\ref{strong-gen}, they do hinder Theorem~\ref{packing-thm} (more specifically they hinder Theorem~\ref{ball-go-down}).

We leave open the question of a possible augmentation of the \RPM\ structure, or a better analysis, that can simultaneously give the following guarantees :
\begin{enumerate}
	\item \textbf{Bounded Depth} : depth of the tree should be $o(n)$, preferably $(\log n)^{\OO{1}}$
	\item \textbf{Packing Guarantee} : of the form $\left(\frac{R}{r}\right)^{\left(d\log \frac{R}{r}\right)^{\OO{1}}}$
	\item \textbf{Space Partitioning Guarantee} : assured size reduction by factor $s$ in $(d\log s)^{\OO{1}}$ levels
\end{enumerate}

\subsubsection*{Acknowledgments}
The authors thank James Lee for pointing out an incorrect usage of the term Assouad dimension in a previous version of the paper. Purushottam Kar thanks Chandan Saha for several fruitful discussions and for his help with the proofs of the Theorems~\ref{strong-gen}~and~\ref{packing-thm}. Purushottam is supported by the Research~I Foundation of the Department of Computer Science and Engineering, IIT Kanpur.

\begin{small}
\bibliography{ref}
\end{small}

\appendix

\end{document}